\documentclass[journal]{IEEEtran}

\usepackage{graphicx}
\usepackage{citesort}
\usepackage{amssymb}
\usepackage{amsfonts}
\usepackage{amsmath}
\usepackage{epsfig}
\usepackage{color}
\usepackage{fancybox}
\usepackage{textcomp}
\usepackage{multirow}
\usepackage{setspace}
\usepackage{psfrag}
\usepackage[ruled,vlined,linesnumbered]{algorithm2e}

\newtheorem{Proposition}{Proposition}

    \newcommand{\qx}{{\bf x}}
    \newcommand{\qy}{{\bf y}}
    \newcommand{\qz}{{\bf z}}

    \newcommand{\qA}{{\bf A}}

    \newcommand{\qJ}{{\bf J}}

    \newcommand{\bbC}{{\mathbb C}}

    \newcommand{\calF}{{\mathcal F}}

    \newcommand{\calN}{{\mathcal N}}

    \newcommand{\Ex}{{\sf E}}

    \newcommand{\Extr}{\operatornamewithlimits{\sf Extr}}
    \newcommand{\mmse}{{\sf mmse}}

    \newcommand{\rmd}{{\rm d}}
    \newcommand{\rmD}{{\rm D}}

\begin{document}

\title{Analysis of Compressed Sensing\\with Spatially-Coupled Orthogonal Matrices}

\author{Chao-Kai Wen
        and~Kai-Kit Wong

\thanks{C.-K. Wen is with the Institute of Communications Engineering, National Sun Yat-sen University, Kaohsiung, Taiwan. E-mail: {\sf ckwen@ieee.org}.}
\thanks{K.-K. Wong is with the Department of Electronic and Electrical Engineering, University College London, London, WC1E 7JE, United Kingdom. E-mail: {\sf kai-kit.wong@ucl.ac.uk}.}
}

\maketitle

\begin{abstract}
Recent development in compressed sensing (CS) has revealed that the use of a special design of measurement matrix, namely the spatially-coupled matrix, can achieve the information-theoretic
limit of CS. In this paper, we consider the measurement matrix which consists of the spatially-coupled \emph{orthogonal} matrices. One example of such matrices are the randomly selected discrete
Fourier transform (DFT) matrices. Such selection enjoys a less memory complexity and a faster multiplication procedure. Our contributions are the replica calculations to find the
mean-square-error (MSE) of the Bayes-optimal reconstruction for such setup. We illustrate that the reconstruction thresholds under the spatially-coupled orthogonal and Gaussian ensembles are
quite different especially in the noisy cases. In particular, the spatially coupled orthogonal matrices achieve the faster convergence rate, the lower measurement rate, and the reduced MSE.
\end{abstract}

\section*{\sc I. Introduction}
Compressed sensing (CS) is a signal processing technique that aims to reconstruct a sparse signal with a higher dimension $(N)$ space from an underdetermined lower dimension $(M)$ measurement
space, with the measurement ratio $\alpha = M/N$ as small as possible. In the literature, the $\ell_{1}$-norm minimization is the most widely used scheme in signal reconstruction because the
$\ell_{1}$-norm minimization is convex and hence can be solved very efficiently \cite{Candes-05IT,Donoho-06IT,Candes-06IT}. However, the measurement ratio of the $\ell_{1}$-reconstruction for a
perfect reconstruction is required to be sufficiently larger than the information-theoretic limit \cite{Donoho-09JAMS,Kabashima-09JSM}.

If the probabilistic properties of the signal are known, then the probabilistic Bayesian inference offers the optimal reconstruction in the minimum mean-square-error (MSE) sense, but the optimal
Bayes estimation is not computationally tractable. By using belief propagation (BP), an efficient and less complex alternative, referred to as approximate message passing (AMP)
\cite{Donoho-09PNAS,Rangan-10ArXiv,Krzakala-12JSM}, has recently emerged. A remarkable result by Krzakala {\em et al.}~\cite{Krzakala-12PRX,Krzakala-12JSM} showed that a sparse signal can be
recovered up to its information theoretical limit if the measurement matrix has some special structure, namely spatially-coupled.\footnote{The idea of spatial coupling was first introduced in
the CS literature by \cite{Kudekar-10Allerton}, where some limited improvement in performance was observed.} Roughly speaking, spatially-coupled matrices are random matrices with a band diagonal
structure as shown in Fig.~\ref{fig:SpatiallyCoupledMatrix}. The authors of \cite{Krzakala-12PRX,Krzakala-12JSM} support this claim using an insightful statistical physics argument. This claim
has been proven rigorously by \cite{Donoho-12ISIT}.

Though AMP is less complex than the Bayes-optimal approach, the implementation of AMP will become prohibitively complex if the size of the signal is very large. This is not only because AMP
still requires many matrix multiplications up to order of $O(MN)$ but also because it requires many memory to store the measurement matrix. It is therefore of great interest to consider some
special measurement matrix permitting faster multiplication procedure and less memory complexity. Randomly selected discrete Fourier transform (DFT) matrices are one such example
\cite{Barbier-13ArXiv,Javanmard-12ISIT}. Using DFT as the measurement matrix, fast Fourier transform (FFT) can be used to perform the matrix multiplications down to the order of $O(N\log_2N)$
and the measurement matrix is not required to be stored. However, in contrast to the case with random matrices with independent entries, there are only a few studies on the performance of CS for
matrices with row-orthogonal ensemble \cite{Kabashima-09JSM,Javanmard-12ISIT,Tulino-13IT,Kabashima-12JSM,Kabashima-14ArXiv,Barbier-13ArXiv,Vehkapera-arXiv13}.

Analysis in \cite{Kabashima-09JSM} revealed that the $\ell_{1}$-reconstruction thresholds are the same under all measurement matrices that are sampled from the rotationally invariant matrix
ensembles. In addition, along the line of $\ell_{1}$-reconstruction, the authors in \cite{Kabashima-12JSM} showed that the gain in performance using concatenation of random orthogonal matrices
is only specific to signal with non-uniform sparsity pattern. In a different context, \cite{Tulino-13IT} also showed that a general class of free random matrices incur no loss in the noise
sensitivity threshold if optimal decoding is adopted. The empirical study in \cite{Donoho-09PNAS} illustrated that the reconstruction ability of AMP is universal with respect to different matrix
ensembles. Furthermore, by empirical studies, \cite{Barbier-13ArXiv} found that using DFT matrices does alter the AMP performance but will not affect the final performance significantly. With
these studies, one might conclude that the reconstruction ability would be nearly universal with respect to different measurement matrix ensembles. However, counter evidences appear recently
when the measurement is corrupted by additive noise \cite{Vehkapera-arXiv13,Kabashima-14ArXiv}. They argued the superiority of the row-orthogonal ensembles over independent Gaussian ensembles in
noisy setting.\footnote{In fact, the significance of orthogonal matrices under other problems (e.g., CDMA and MIMO) was pointed out in \cite{Takeda-06EPL,Hatabu-09PRE}.} Therefore, it is not
fully understood how the measurement matrix with row-orthogonal ensemble affects the CS performance.

In this paper, our aim is to fill this gap by investigating the MSE in the optimal Bayes inference of the sparse signals if the measurement matrix consists of the spatially-coupled orthogonal
matrices. In particular, by using the replica method, we get the state evolution of the MSE for CS with spatially-coupled orthogonal matrices. Based on the derived state evolution, we are able
to observe closer behaviors regarding the CS with orthogonal matrices. Several interesting observations will be made via the statistical physics argument in \cite{Krzakala-12JSM}.

As a summary, our finding is that the reconstruction thresholds under the row orthogonal and i.i.d.~Gaussian ensembles are quite different especially in noisy scenarios. The construction
thresholds seem universal only in a very low noise variance regime. In the higher noisy variance, the reconstruction thresholds of the row-orthogonal ensemble are significantly lower than those
of the i.i.d.~Gaussian ensemble. In addition, we notice that the case with the spatially coupled orthogonal matrices enjoys 1) the faster convergence rate, 2) the lower measurement rate, and 3)
the lower MSE result.

\section*{\sc II. Problem Formulation}

\begin{figure}
\begin{center}
\resizebox{2.75in}{!}{%
\includegraphics*{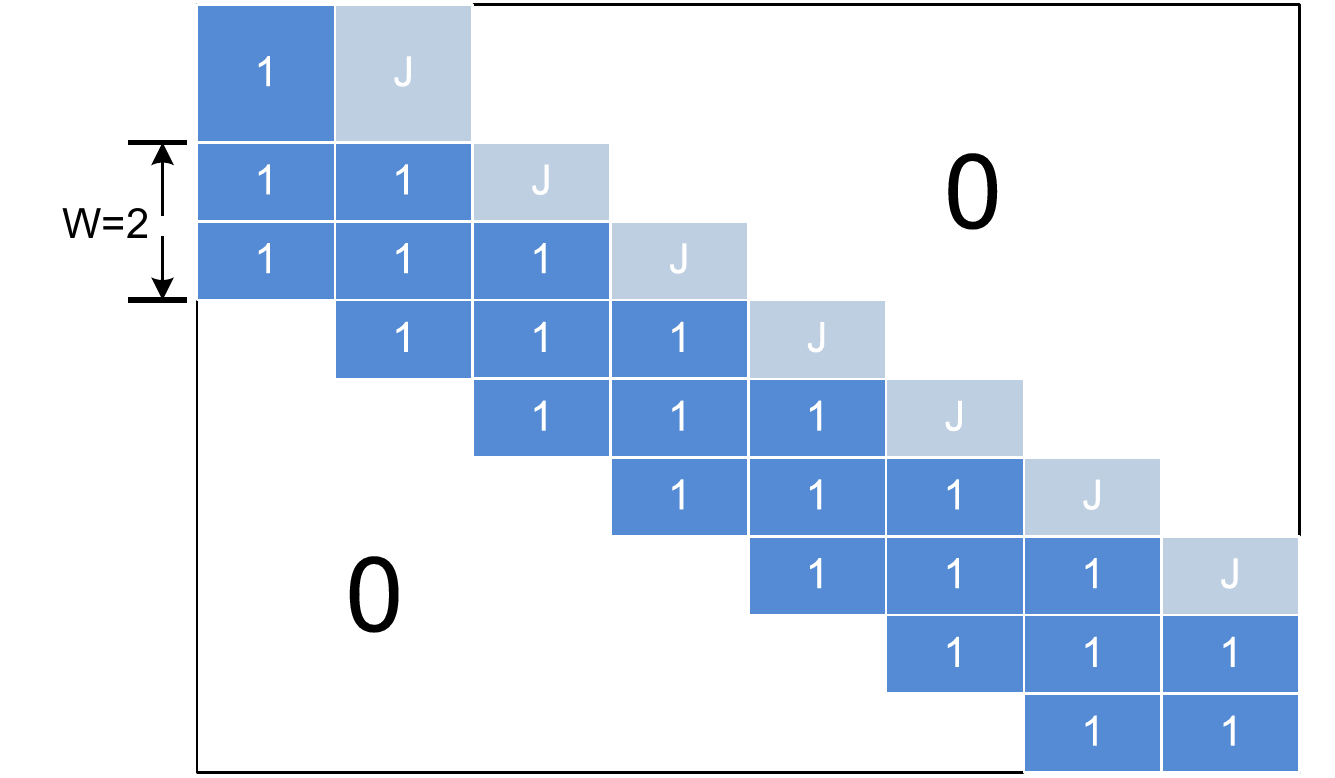} }%
\caption{An example of spatially-coupled matrix introduced in \cite{Krzakala-12JSM,Barbier-13ArXiv}, where $L_c=L_r = 8$. Each block is obtained by randomly selecting and re-ordering from the DFT matrix. The upper diagonal blocks have components with variance $J$. For readers' convenience, the similar notations to those in \cite{Krzakala-12JSM} are used to demonstrate the spatially-coupled matrix.}\label{fig:SpatiallyCoupledMatrix}
\end{center}
\end{figure}

We consider the noisy CS problem
\begin{equation}\label{eq:sysModel}
    \qy= \qA\qx + \sigma\qz,
\end{equation}
where $\qy \in \bbC^{M}$ is a measurement vector, $\qA \in \bbC^{M \times N}$ denotes a known measurement matrix, $\qx \in \bbC^{N}$ is a signal vector, $\qz \in\bbC^{M}$ is the standard
Gaussian noise vector, and $\sigma$ represents the noise magnitude. We denote by $\alpha = M/N$ the measurement ratio (i.e., the number of measurements per variable).

The spatially-coupled matrix $\qA$ used in this paper follows that in \cite{Krzakala-12JSM}, see Fig.~\ref{fig:SpatiallyCoupledMatrix}. The $N$ components of the signal vector $\qx$ are split
into $L_c$ blocks of $N_p$ variables for $p=1,\dots,L_c$. We denote $\gamma_p = N_p/N$. Next, we split the $M$ components of the measurements $\qy$ into $L_r$ blocks of $M_q$ measurements, for
$q=1,\dots,L_r$. As a result, $\qA$ is composed of $L_r\times L_c$ blocks. Each block $\qA_{q,p} \in \bbC^{M_q \times N_p}$ is obtained by randomly selecting and re-ordering from the standard
DFT matrix multiplied by $\sqrt{J_{q,p} \gamma_p}$.\footnote{The standard DFT matrix has been normalized by $1/\sqrt{N_p}$. However, we assume that the components of $\qA_{q,p}$ have variance
$J_{q,p}/N$. To this end, the factor $\gamma_p = N_p/N$ is used to adjust the normalization in each block.} The measurement ratio of $(q,p)$-group is $\alpha_{q,p} = M_q/N_p$. We have an $L_r
\times L_c$ coupling matrix $\qJ \triangleq [J_{q,p}]$.

CS aims to reconstruct $\qx$ from $\qy$. We suppose that each entry of $\qx$ is generated from a distribution $P(x)$ independently. In particular, the signals are sparse where the fraction of
non-zero entries is $\rho$ and their distribution is $ g$. That is,
\begin{equation}\label{eq:px}
    P(\qx) = \prod_{n=1}^{N} P(x_n) = \prod_{n=1}^{N} \Big( (1-\rho)\delta(x_n) + \rho g(x_n) \Big).
\end{equation}
The Bayes optimal way of estimating $\qx$ that minimizes the MSE, defined as $E \triangleq \sum_{n=1}^{N} |\hat{x}_n - x_n|/N$, is given by \cite{Poor-94BOOK}
\begin{equation} \label{eq:estx}
    \hat{\qx} \triangleq \Ex\{\qx|\qy\} = \int \qx p(\qx|\qy) \rmd\qx ,
\end{equation}
where $p(\qx|\qy)$ is the posterior probability of $\qx$ given observation of $\qy$. Following Bayes theorem, we have
\begin{equation}
    p(\qx|\qy) = \frac{p(\qy|\qx) p(\qx)}{p(\qy)},
\end{equation}
where the conditional distribution of $\qx$ given $\qy$ in (\ref{eq:sysModel}) is
\begin{equation}
    p(\qy|\qx) = \frac{1}{(\pi \sigma^2)^M} e^{-\frac{1}{\sigma^2} \| \qy - \qA\qx \|^2}.
\end{equation}
Our aim is to study the MSE in the optimal Bayes inference.

\section*{\sc III. Analytical Result}
Before proceeding, it is useful to understand the posterior mean estimator (\ref{eq:estx}) by revisiting a scalar single measurement
\begin{equation} \label{eq:ScalarSysModel}
    y= x +  \varsigma^{-\frac{1}{2}}  z.
\end{equation}
This is a special case of (\ref{eq:sysModel}) with $M=N=1$. According to (\ref{eq:estx}), MMSE is achieved by the conditional expectation
\begin{equation} \label{eq:ScalarEstX}
    \Ex\{x|y\} = \int  x p(x|y) \rmd x,
\end{equation}
where $p(x|y) = p(y|x) p(x)/p(y)$ and $p(y|x) =\frac{\varsigma}{\pi} e^{-\varsigma |y-x|^2}$. Note that $\hat{x}(y)$ changes with $y$ while we will suppress $y$ for brevity. Finally, we define
$\mmse(\cdot)$ of this setting as
\begin{equation} \label{eq:defMMSE}
    \mmse(\varsigma) \triangleq \Ex\left\{ \left|x-\Ex\{x|y\}\right|^2 \right\},
\end{equation}
in which the expectation is taken over the joint conditional distribution $p(y,x) = p(y|x) p(x)$.

Explicit expressions of $\mmse$ are available for some special signal distributions. For example, if the signal distribution $p(x)$ follows the Bernoulli-Gaussian (BG) density, i.e., $ g$ is the
standard complex Gaussian distribution, then we have
\begin{equation} \label{eq:ScalarEstX2}
\Ex\{x|y\} = \frac{\rho\calN(y;1+\varsigma^{-1})}{\rho\calN(y;1+\varsigma^{-1})+(1-\rho)\calN(y;\varsigma^{-1})}\frac{y}{1+\varsigma^{-1}},
\end{equation}
where $\calN(y;c)$ denotes a Gaussian probability density function (pdf) with zero mean and variance $c$, i.e., $\calN(y;c) \triangleq 1/(\pi c)e^{-|y|^2/c}$. Then we obtain explicitly
\begin{equation} \label{eq:defMMSE2}
    \mmse(\varsigma) = \rho - \frac{\rho^2\varsigma}{\varsigma+1} \int \rmD z \frac{|z|^2}{\rho+(1-\rho)e^{-|z|^2\varsigma}(\varsigma+1)},
\end{equation}
where $\rmD z \triangleq \frac{\rmd \Re{z} \rmd \Im{z} }{\pi} e^{-|z|^2}$ with $\Re{z}$ and $\Im{z}$ being the real and imaginary parts of $z$, respectively.

Although the analytical result of $\mmse$ of the scalar measurement is available, the task of obtaining the corresponding result to the vector case (\ref{eq:sysModel}) might appear daunting.
Surprisingly, tools from statistical mechanics enable such development in large system limits. The key for finding the statistical properties of (\ref{eq:estx}) is through the average free
entropy \cite{Krzakala-12JSM}
\begin{equation}\label{eq:FreeEn}
    \calF \triangleq \frac{1}{N}\Ex_{\qy,\qA}\left[\log Z(\qy,\qA)\right],
\end{equation}
where
\begin{equation} \label{eq:ZyH}
Z(\qy,\qA)\triangleq\Ex_{\qx}\left[e^{-\frac{1}{ {\sigma}^2}\left\|\qy-  \qA \qx \right\|^2}\right]
\end{equation}
is the partition function. The similar approach also has been used under different settings, e.g., \cite{Krzakala-12JSM,Tulino-13IT,Kabashima-12JSM,Vehkapera-arXiv13}. The analysis of
(\ref{eq:FreeEn}) is still difficult. The major difficulty in (\ref{eq:FreeEn}) lies in the expectations over $\qy$ and $\qA$. We can, nevertheless, greatly facilitate the mathematical
derivation by rewriting $\calF$ as
\begin{equation}\label{eq:LimF}
\calF = \frac{1}{N} \lim_{r\rightarrow 0}\frac{\partial}{\partial r}\log\Ex_{\qy,\qA}\left[Z^r(\qy,\qA)\right],
\end{equation}
in which we have moved the expectation operator inside the log-function. We first evaluate $\Ex_{\qy,\qA}\left[Z^r(\qy,\qA)\right]$ for an integer-valued $r$, and then generalize it for any
positive real number $r$. This technique is called the replica method \cite{Edwards-75JPF}, and has been widely adopted in the field of statistical physics \cite{Nishimori-01BOOK}.

In the analysis, we use the assumptions that $N_p \rightarrow \infty$, for all $p=1,\ldots,L_c$, and $M_q \rightarrow \infty$, for all $q=1,\ldots,L_r$, while keeping $M_q/N_p = \alpha_{q,p}$
fixed and finite. For convenience, we refer to this large dimensional regime simply as $N \rightarrow \infty$.

Under the assumption of replica symmetry, the following results are obtained.

\begin{Proposition} \label{Pro1}
As $N \rightarrow \infty$, the free entropy is given by
\begin{multline}\label{eq:GenFree}
\calF(\{\varsigma_{q,p},\varepsilon_{q,p}\}) = \sum_{p=1}^{L_c}\gamma_{p} \Ex_{y_p}\left\{ \log \Ex_{x_p}\left\{ e^{-\varsigma_{p}| y_{p} - x_{p} |^2} \right\} \right\} \\
  + \sum_{q=1}^{L_r}\sum_{p=1}^{L_c}\gamma_{p}\varepsilon_{q,p}\varsigma_{q,p}
  + \sum_{q=1}^{L_r} G(\{\varepsilon_{q,p}\}) + (1-\alpha),
\end{multline}
where $\varsigma_{p} \triangleq \sum_{q=1}^{L_r} \varsigma_{q,p}$, the outer expectation $\Ex_{y_p}\{ \cdot \}$ is taken over the joint conditional distribution
\begin{equation}
    p(y_p,x'_p) = \frac{\varsigma}{\pi} e^{-\varsigma|y_p - x'_p|^2} p(x'_p),
\end{equation}
and
\begin{multline}
  G(\{\varepsilon_{q,p}\})
  \triangleq  \Extr_{\{\Lambda_{q,p}\}} \Bigg\{ - \alpha_{q,p}\gamma_p \log\left(1 + \sum_{p=1}^{L_c} \frac{\gamma_p J_{q,p}}{\sigma^2\Lambda_{q,p}} \right)
  \\ + \sum_{p=1}^{L_c} \gamma_{p} (\Lambda_{q,p} \varepsilon_{q,p} - \log\Lambda_{q,p} \varepsilon_{q,p}-1)  \Bigg\},
\end{multline}
where $\Extr_{X}\{\cdots\}$ denotes the extremization with respect to $X$. The quantities $\{\varsigma_{q,p},\varepsilon_{q,p}\}$ are chosen to maximize (\ref{eq:GenFree}).
\end{Proposition}

\begin{proof}
The proposition can be obtained by applying the techniques in \cite{Takeda-06EPL,Vehkapera-arXiv13} after additional manipulations.
\end{proof}

\begin{Proposition} \label{Pro_StateEvaluation}
The asymptotic evolution of the MSE $\varepsilon_p$ in each block $p$ is given by
\begin{equation} \label{eq:evolutionMSE}
    \varepsilon_p^{(t)} = \mmse\left( \sum_{q=1}^{L_r} \varsigma_{q,p}^{(t-1)} \right)
\end{equation}
where
\begin{subequations} \label{eq:sdPoint2}
\begin{align}
\Lambda_{q,p}^{(t)} &= \frac{1}{\varepsilon_p^{(t)}} - \varsigma_{q,p}^{(t-1)}, \label{eq:fixPoiny_Lambda} \\
\Delta_{q,p}^{(t)} &\triangleq \frac{ \alpha_{q,p}\frac{\gamma_p J_{q,p} }{\Lambda_{q,p}^{(t)}}}{\sigma^2 + \sum_{l=1}^{L_c} \frac{\gamma_l J_{q,l}}{\Lambda_{q,l}^{(t)}} } . \\
\varsigma_{q,p}^{(t)} &= \frac{\Lambda_{q,p}^{(t)}\Delta_{q,p}^{(t)}}{1-\Delta_{q,p}^{(t)}}. \label{eq:fixPoiny_varsigma}
\end{align}
\end{subequations}
As $t \rightarrow \infty$ (i.e., in thermodynamic), $\{\varsigma_{q,p},\Lambda_{q,p},\Delta_{q,p}\}$ will converge to values which locally maximize the free entropy.
\end{Proposition}

\begin{proof}
The state evolution of the Bayes-optimum corresponds to the steepest ascent of the free entropy (\ref{eq:GenFree}).
\end{proof}

If the ensembles of $\{\qA_{q,p}\}$ are Gaussian, the free entropy shares the same form as (\ref{eq:GenFree}) while $G$ should be replaced by
\begin{equation}
  G_{\sf Gaussian}(\{\varepsilon_{q,p}\})
  \triangleq - \alpha_{q,p}\gamma_p \log\left(1 + \sum_{p=1}^{L_c} \frac{\gamma_p J_{q,p}\varepsilon_{q,p}}{\sigma^2} \right) .
\end{equation}
In that case, the evolution of the MSE $\varepsilon_p$ in each block $p$ also follows the same form as (\ref{eq:evolutionMSE}) while $\varsigma_{q,p}^{(t)}$ should be \cite{Krzakala-12JSM}
\begin{equation}
  \varsigma_{q,p}^{(t)} = \frac{ \alpha_{q,p} \gamma_p J_{q,p} }{\sigma^2 + \sum_{l=1}^{L_c} \gamma_l J_{q,l}\varepsilon_{q,p} }.
\end{equation}
It is evident that the evolution of the MSE for row-orthogonal matrices instead of random ones indeed alters the performance of the Bayes-optimal reconstruction.

\section*{\sc IV. Discussions}
To better understand the relation between the free entropy and the MSE using the Bayes-optimal reconstruction, let us first consider the simplest case without the spatially-coupled matrix, i.e.,
$L_c=L_r=1$ and $J_{1,1} = 1$. For notational convenience, we refer $\varsigma_{1,1}$ and $\varepsilon_{1,1}$ to $\varsigma$ and $\varepsilon$, respectively. In Fig.~\ref{fig:FreeEnergy}, we
plot the free entropy $\calF(\varepsilon)$ for a signal of density $\rho = 0.4$, variance of noise $\sigma^2 = 10^{-4}$, and different values of the measurement rate $\alpha$. In maximizing
$\calF(\varepsilon)$ with respect to $\varepsilon$, one obtains the Bayes-optimal MSE. In particular, the state evolution in Proposition \ref{Pro_StateEvaluation} performs a steepest ascent in
$\calF(\varepsilon)$ and gets trapped at one of the local maximas according to the initial value of $\varepsilon$. Nonetheless, the only initialization that is algorithmically possible, e.g.,
BP, is starting from a large value, e.g., $\varepsilon^{(0)} = \rho$ (or larger). This means that BP converges to the much higher MSE than the MSE corresponding to the global maximum of the free
entropy. We can observe a phenomenon similar to that describes in \cite{Krzakala-12JSM} for Gaussian i.i.d.~matrices. Following \cite{Krzakala-12JSM}, we define three phase transitions as
follows:
\begin{itemize}
\item $\alpha_{\sf d}$ is defined as the \emph{largest} $\alpha$ for which $\calF(\varepsilon)$ has two local maximas.
\item $\alpha_{\sf s}$ is defined as the \emph{smallest} $\alpha$ for which $\calF(\varepsilon)$ has two local maximas.
\item $\alpha_{\sf c}$ is defined as the value $\alpha$ for which the two maximas of $\calF(\varepsilon)$ are identical.
\end{itemize}

\begin{figure}
\begin{center}
\resizebox{3.25in}{!}{%
\includegraphics*{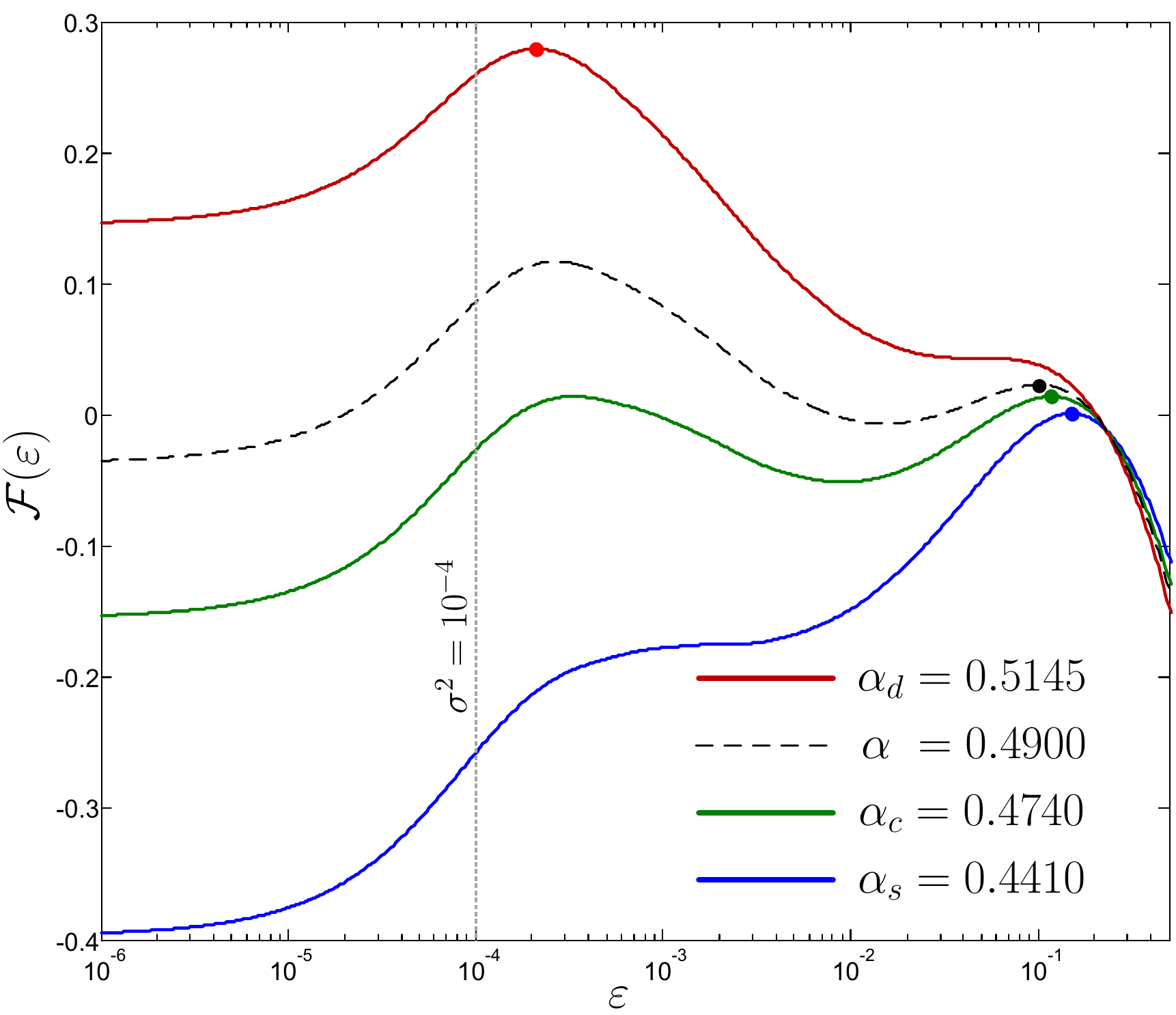} }%
\caption{The free entropy $\calF(\varepsilon)$ for $L_c=L_r=1$, measurement rate $\rho = 0.4$ and noise level $\sigma^2 = 10^{-4}$ under different value of density $\rho$. Marks correspond to the free entropy which in principle can be achieved by the BP algorithm.}\label{fig:FreeEnergy}
\end{center}
\end{figure}

From Fig.~\ref{fig:FreeEnergy}, we see that $\alpha_{\sf d}> \alpha_{\sf c} > \alpha_{\sf s}$. If $\alpha > \alpha_{\sf d}$, then the global maximum of $\calF(\varepsilon)$ is the only maximum
which is at the smaller value of MSE comparable to $\sigma^2$. This means that BP is able to reach the good MSE. However, when $\alpha_{\sf d} > \alpha > \alpha_{\sf c}$, the BP algorithm
converges to the right-most local maxima. The global maximum of $\calF(\varepsilon)$ is no longer reached by the BP algorithm although it corresponds to a small value of MSE. Hence, BP is
suboptimal in this region. Finally, if the sampling rate keeps decreasing, the global maximum of $\calF(\varepsilon)$ appears at the higher MSE and BP will surely converge to it.

As a consequence, there is a BP threshold $\alpha_{\sf d}$ below which the algorithm fails to achieve the good MSE. However, the remarkable result of \cite{Krzakala-12PRX} is that when the
concept of spatial coupling is used to design the measurement matrix $\qA$ then the BP algorithm reconstructs successfully down to $\alpha_{\sf c}$.

\begin{figure}
\begin{center}
\resizebox{3.5in}{!}{%
\includegraphics*{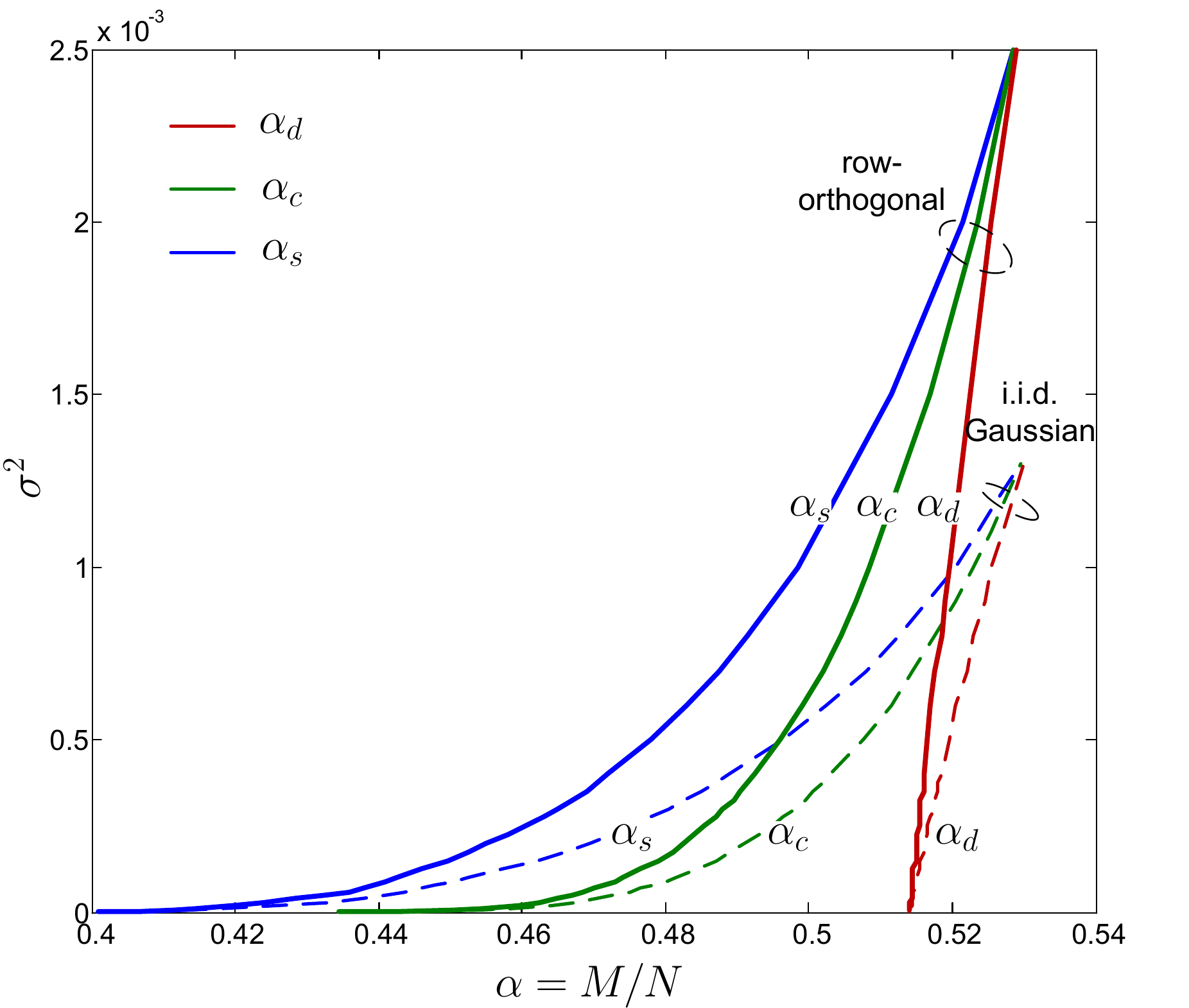} }%
\caption{The three phase transition lines for row-orthogonal and i.i.d.~Gaussian ensembles in noisy measurements when the signal density is $\rho = 0.4$.}\label{fig:Phase}
\end{center}
\end{figure}

From the above, we have seen three kinds of phase transition behavior as a function of $\alpha$ for a given $\sigma^2$. The three phase transitions are expected to depend upon $\sigma^2$. In
Fig.~\ref{fig:Phase}, we plot the dependence of $\alpha_{\sf d}$, $\alpha_{\sf c}$, and $\alpha_{\sf s}$ on the noise variance under row-orthogonal and i.i.d.~Gaussian ensembles. For the larger
noise variance, both of the ensembles appear to have no such sharp threshold among $\alpha_{\sf d}$, $\alpha_{\sf c}$, and $\alpha_{\sf s}$. We observe that in very low noise variance, the phase
transition lines between row-orthogonal and i.i.d.~Gaussian ensembles are extremely closed. Based on this observation and the knowledge that for the noise-free case, we conclude that the
construction threshold seems universal only with very low noise variance. Nonetheless, with higher noisy variance, we can see the superiority of the row-orthogonal ensembles to the
i.i.d.~Gaussian ensembles in terms of the following three perspectives. First, the BP threshold $\alpha_{\sf d}$ of the row-orthogonal ensemble is lower than that of the i.i.d.~Gaussian
ensemble. Secondly, for region of $\sigma^2 \gtrsim 0.0013$, there is no sharp phase transition under the i.i.d.~Gaussian ensemble, while the region is extended to $\sigma^2 \gtrsim 0.0025$
under the row-orthogonal ensemble. Finally, the $\alpha_{\sf c}$ transition line under the row-orthogonal ensemble is much lower than that under the i.i.d.~Gaussian ensemble. This implies that
with a proper spatially-coupled matrix $\qA$, the BP algorithm under the row-orthogonal ensemble can achieve the good MSE down to the lower measurement rate in the noisy case.

In Fig.~\ref{fig:MSE}, we plot the MSE achieved by the BP algorithm as a function of the noise variance $\sigma^2$ when the measurement rates are $\alpha_{\sf d}$ in Fig.~\ref{fig:Phase}. As can
be observed, the row-orthogonal ensemble even achieves the lower MSE than the i.i.d.~Gaussian ensemble. This together with the results from Fig.~\ref{fig:Phase} indicates that the row-orthogonal
ensemble not only allows the lower measurement rate but also achieves a lower MSE.

For spatially coupled matrices, Fig.~\ref{fig:MSEVsIteration} compares the evaluation of the MSE using Proposition \ref{Pro_StateEvaluation}. We observe that the evaluation of the MSE with the
spatially coupled orthogonal matrices converges faster than that with the spatially coupled Gaussian matrices. The same phenomenon can also be observed from \cite[Fig.~3]{Barbier-13ArXiv}, in
which a practical AMP algorithm was employed to perform the signal reconstructions.

It is known from \cite{Krzakala-12PRX} that with the spatially-coupled Gaussian matrices, the good MSE can be achieved down to $\alpha_{\sf c}$. Therefore, it is useful to examine whether the
spatially coupled row-orthogonal matrices can achieve the good MSE under the lower measurement rate than the spatially coupled Gaussian matrices. To have such examination, we conduct a numerical
analysis in the following setting: 1) $W = 2$, $\alpha_{\rm seed} = 0.70$, $\alpha_{\rm bulk} = 0.489$, and $J= 2.5$ for Gaussian ensembles; 2) $W = 2$, $\alpha_{\rm seed} = 0.70$, $\alpha_{\rm
bulk} = 0.484$, and $J= 1.5$ for row-orthogonal ensemble. As $L$ increases in both cases, the measurement rate $\alpha$ decreases to $\alpha_{\rm bulk}$. In fact, with more numerical results, we
indeed observe that the good MSE can be achieved to $\alpha \rightarrow \alpha_{\rm bulk} \approx \alpha_{\sf c}$. In addition, we observe that the spatially coupled row-orthogonal matrices
enjoy 1) the faster convergence rate, 2) the lower measurement rate, and 3) the lower MSE result.

Finally, we remark that a way to design a spatially-coupled matrix with row-orthogonal ensemble should be different from that with Gaussian ensemble especially in noisy cases. The development of
the AMP algorithm that corresponds to the evolution of MSE in Proposition \ref{Pro_StateEvaluation} is under way.

\begin{figure}
\begin{center}
\resizebox{3.5in}{!}{%
\includegraphics*{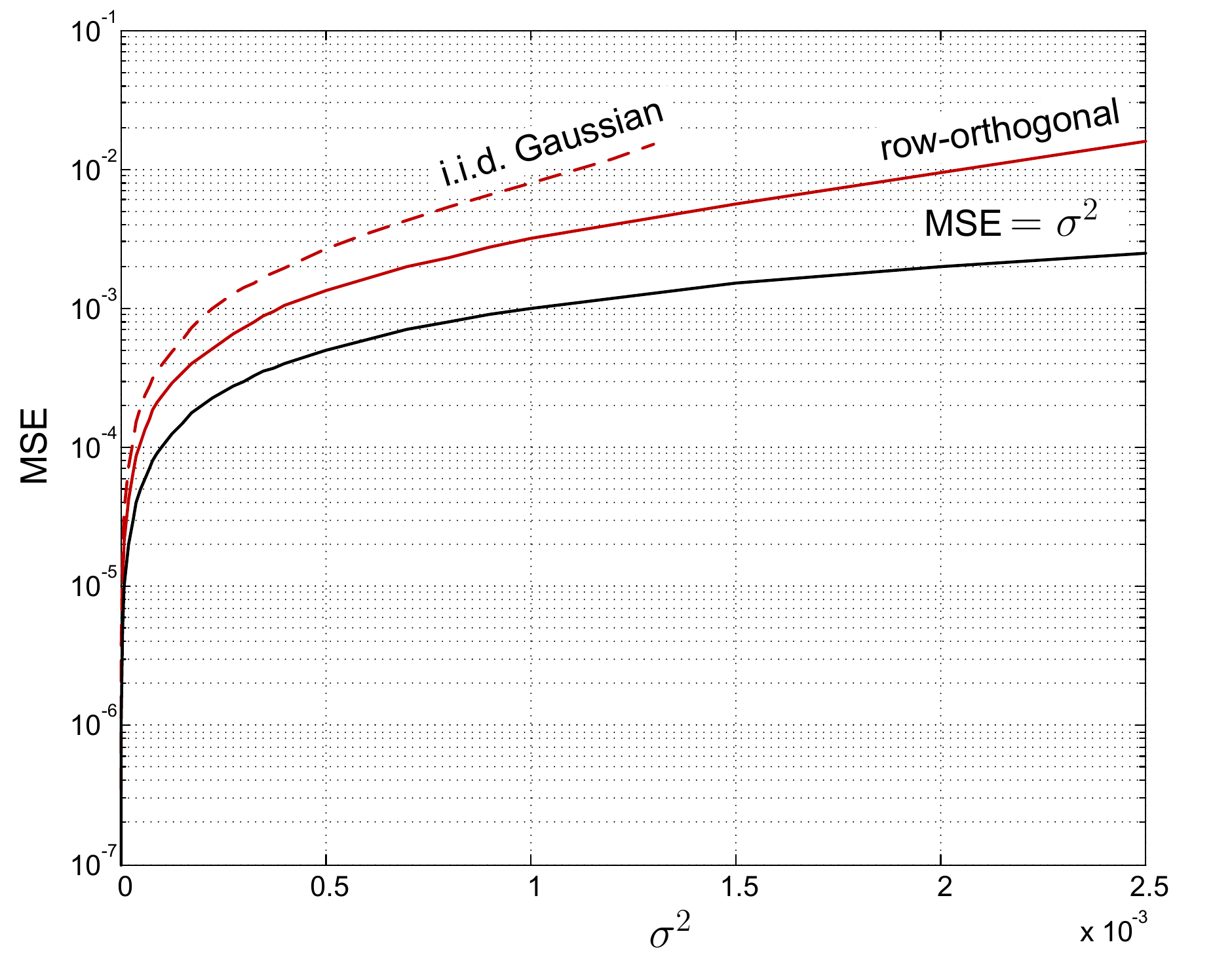} }%
\caption{The MSE achieved by the BP algorithm versus the noise variance $\sigma^2$. The measurement rate is at the BP threshold $\alpha_{\sf d}$. The block solid line (lowermost) marks the noise level which is the re-constructibility limit.}\label{fig:MSE}
\end{center}
\end{figure}

\begin{figure}
\begin{center}
\resizebox{3.5in}{!}{%
\includegraphics*{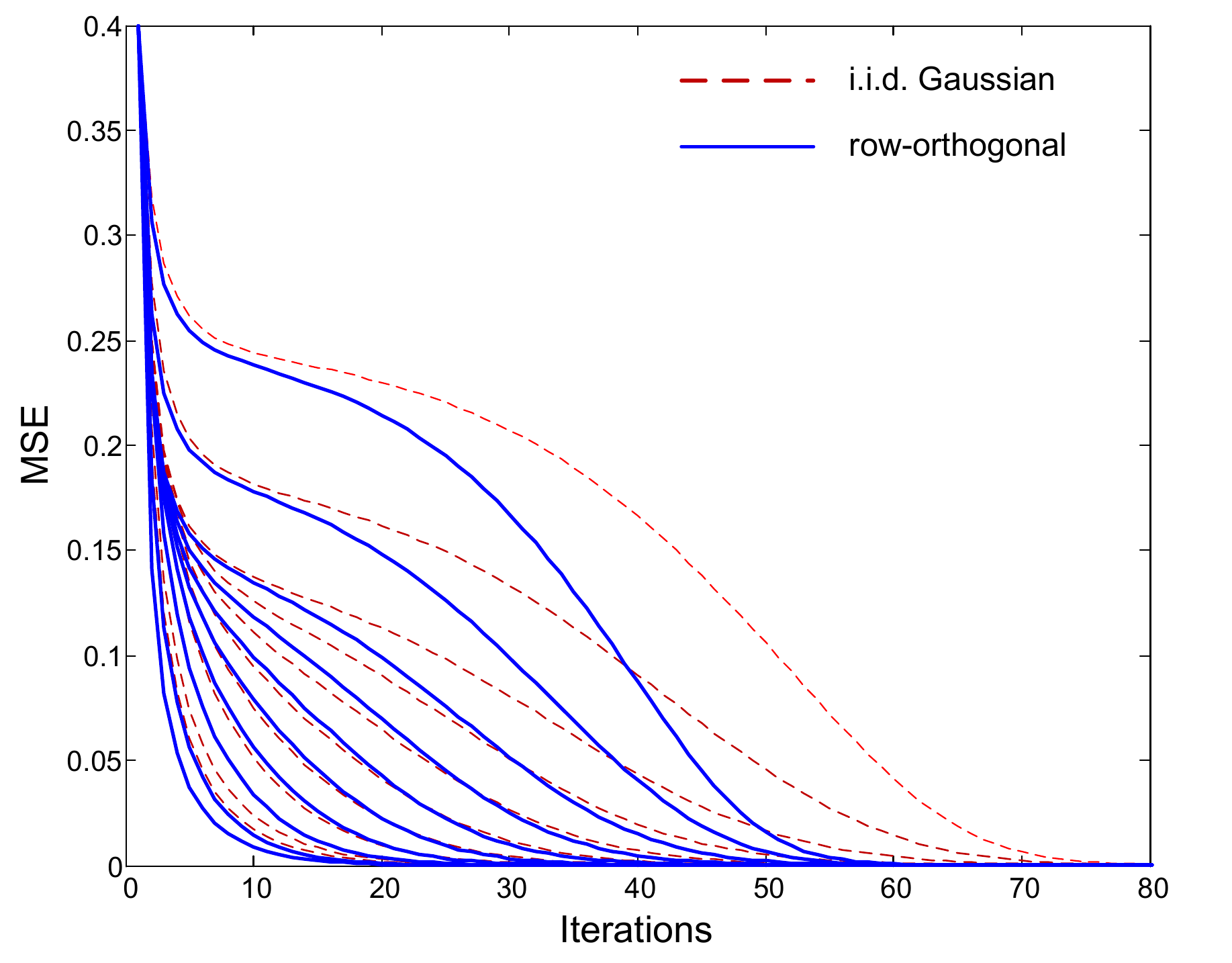} }%
\caption{Evolution of the MSE in each block in the noisy case with $\sigma^2 = 10^{-6}$. We use the seeding matrix with $W = 2$, $L=10$, $\alpha_{\rm seed} = 0.70$, $\alpha_{\rm bulk} = 0.49$, and $J= 0.5$.}\label{fig:MSEVsIteration}
\end{center}
\end{figure}

\section*{V. Conclusion}

We have derived the MSE in the optimal Bayes inference of the sparse signals if the measurement matrix consists of the spatially-coupled orthogonal matrices. The analysis provides a step towards
understanding of the behaviors of the CS with orthogonal matrices. In particular, the numerical results have revealed that the spatially coupled row-orthogonal matrices enjoy the faster
convergence rate, the lower measurement rate, and the lower MSE result. In addition, we remark that a way to design a spatially-coupled matrix with row-orthogonal ensemble should be different
from that with Gaussian ensemble especially in noisy cases. The derived results in this paper can serve as an efficient way to design spatially-coupled orthogonal matrices that have good
performance.

{\renewcommand{\baselinestretch}{1.1}
\begin{footnotesize}

\end{footnotesize}}

\end{document}